\documentclass[reqno,10pt]{article}
\usepackage{graphicx}
\usepackage{amsmath}
\usepackage{amssymb}
\usepackage{amsthm}

\newtheorem{theorem}{Theorem}
\newtheorem{lemma}[theorem]{Lemma}
\newtheorem{corollary}[theorem]{Corollary}
\theoremstyle{definition}

\newtheorem{remark}[theorem]{Remark}
\numberwithin{equation}{section}
\numberwithin{theorem}{section}

\newenvironment{OMabstract}{\noindent\textbf{Abstract.} }{\medskip}
\newenvironment{OMsubjclass}{\noindent\textbf{Mathematics Subject Classification (2020):} }{\medskip}
\newenvironment{OMkeywords}{\noindent\textbf{Keywords:}  }{\medskip}

\begin{document}

\author{Anna G\l\'{o}wczyk and Sergiusz Ku\.{z}el} 
\title{On the $S$-matrix of Schr\"{o}dinger operator with nonlocal $\delta$-interaction}
\maketitle

\begin{OMabstract}
    Schr\"{o}dinger operators with nonlocal $\delta$-interaction are studied with the use of the Lax-Phillips scattering theory methods. 
    The condition of applicability of the Lax-Phillips approach in terms of non-cyclic functions is established.
    Two formulas for the $S$-matrix are obtained. The first one  deals with the Krein-Naimark resolvent formula  and the Weyl-Titchmarsh function,
whereas the second one is based on modified reflection and transmission coefficients. The $S$-matrix $S(z)$ is analytical in the lower half-plane $\mathbb{C_-}$
when the Schr\"{o}dinger operator with nonlocal $\delta$-interaction is positive self-adjoint. Otherwise, $S(z)$ is a meromorphic matrix-valued function 
in $\mathbb{C_-}$ and its properties are closely related to the properties of the corresponding 
 Schr\"{o}dinger operator. Examples of $S$-matrices are given.   
\end{OMabstract}

\begin{OMkeywords}
     Lax-Phillips scattering scheme,  scattering matrix,  $S$-matrix, nonlocal $\delta$-interaction, non-cyclic function
\end{OMkeywords}

\begin{OMsubjclass}
     47B25, 47A40.
\end{OMsubjclass}

 \section{Introduction}\label{intro}
Theory of non self-adjoint operators  attracts a steady interests in various fields of mathematics and physics,
see, e.g., \cite{BGSZ} and the reference therein.  This interest grew considerably due to the recent progress in
theoretical physics of  pseudo-Hermitian Hamiltonians \cite{Bender}.

In the present paper we study non-self-adjoint
Schr\"{o}dinger operators with \emph{nonlocal}  point
interaction.  Self-adjoint operators have been investigated by Nizhnik et al. \cite{Nizhnik, Nizhnik1, Nizhnik2, Nizhnik1b}.
The case of non-self-adjoint operators with nonlocal point interaction is more complicated and it requires more detailed analysis. 
One of the simplest models of  a non-local $\delta$-interaction is
\begin{equation}\label{new1}
-\frac{d^2}{dx^2}+a<\delta,\cdot>\delta(x)+<\delta,\cdot>q(x)+(\cdot, q)\delta(x)
\quad a\in\mathbb{C},
\end{equation}
where $\delta$ is the delta-function, $q\in{L_2(\mathbb{R})}$, and
$(\cdot,\cdot)$ is the inner product (linear in the first
argument) in $L_2(\mathbb{R})$.
The expression \eqref{new1} determines the following
operator acting in  $L_2(\mathbb{R})$:
\begin{equation}\label{AGH1}
H_{aq}f=-\frac{d^2f}{dx^2}+f(0)q(x),
\end{equation}
\begin{equation}\label{ggg1}
 \mathcal{D}(H_{aq})=\left\{f\in{W}_2^2(\mathbb{R}\backslash\{0\}) \ : \
 \begin{array}{l}
 f_s(0)=0   \\
 f_s'(0)=af_r(0)+(f, q) \end{array} \right\}
\end{equation}
where $f_s(0)=f(0+)-f(0-)$ and $f_r(0)=\displaystyle{\frac{f(0+)+f(0-)}{2}}.$
 
The operator $H_{aq}$ is self-adjoint if and only if $a\in\mathbb{R}$ and 
it can be interpreted as a Hamiltonian corresponding to the
non-local $\delta$-interaction \eqref{new1}. Setting  $q=0$,  we obtain
an operator $H_{a}:=H_{a0}$ generated by the ordinary $\delta$-interaction 
 $$
-\frac{d^2}{dx^2}+a<\delta,\cdot>\delta(x).
$$

The spectral analysis of non-self-adjoint $H_{aq}$ \ ($a\in\mathbb{C}\setminus\mathbb{R}$) was carried out in \cite{KZ}.
One of interesting features is  that non-real  $a$ determines the measure of non-self-adjointness of 
 $H_{aq}$, while the function $q$ is responsible for the appearance of 
 exceptional points and eigenvalues on continuous
spectrum \cite[Example 5.3 and Sec. 6]{KZ}.

In the present paper, we investigate $H_{aq}$ by the scattering theory methods.  
For the case $a=0$, the scattering matrix $S(\delta)$ of $H_{0q}$ was constructed in 
\cite[Sec. 5]{Nizhnik} with the use of modified Jost solutions. 
In contrast to \cite{Nizhnik}  we study the general case $a\in\mathbb{C}$
with the use of an operator-theoretical interpretation of the Lax-Phillips approach in scattering theory \cite{LF} that was
consistently developed in  \cite{CK, KU1, AlAn, KU2}.  We prefer this approach because it involves a simple  
algorithm for an explicit calculation of the analytic continuation\footnote{`The most beautiful and important aspect of the Lax-Phillips approach is
that certain analyticity properties of the scattering operator arise naturally' \cite[p.211]{RS3}} of the scattering matrix into the lower half-plane $\mathbb{C}_-$. 

The paper is organized as follows. We begin with presentation of necessary facts about the Lax-Phillips scattering theory. Further,
in Sec. \ref{sec.3},  we analyze for which operators $H_{aq}$ one can apply  the Lax-Phillips approach.
For technical reasons it is convenient to work with unitary equivalent copies  ${\bf H}_{a{\bf q}}$  of the operators $H_{aq}$ 
acting in the Hilbert space  $L_2(\mathbb{R}_+, \mathbb{C}^2)$, see \eqref{newww6a}, \eqref{newww6}.
The main result (Theorem \ref{yyyy}) implies that  ${\bf H}_{a{\bf q}}$ can be investigated in framework of the
 Lax-Phillips theory under the condition that ${\bf q}$ is non-cyclic with respect to the backward shift operator.
For such kind of positive self-adjoint operators ${\bf H}_{a{\bf q}}$, two formulas of the analytical continuation
$S(z)$ of the scattering matrix $S(\delta)$ into $\mathbb{C_-}$ are obtained in Sec. \ref{sec.4}.
The first one  \eqref{ups4}  deals with the Krein-Naimark resolvent formula \eqref{DDD4} and the Weyl-Titchmarsh function \eqref{DDD3},
whereas the second one  \eqref{AGHHH}  is based on the modified reflection $R_z^i$ and the transmission $T_z^i$ coefficients
that is more familiar for non-stationary scattering theory.  

We mention that the relationship between
scattering matrices and the extension theory subjects like
Krein-Naimark formula and Weyl-Titchmarsh function
was established for various cases \cite{AP,  BMN, CKS} and it provides additional possibilities for the study of scattering systems. 

In Sec \ref{sec.5}, the formula \eqref{ups4} is used for the definition of  $S$-matrix $S(z)$ for each operator ${\bf H}_{a{\bf q}}$
(assuming, of course, that ${\bf q}$ is non-cyclic).  If ${\bf H}_{a{\bf q}}$ is positive self-adjoint, then
the $S$-matrix is the direct consequence of proper arguments of the Lax-Phillips theory and it coincides with the analytical continuation of the Lax-Phillips scattering matrix
into $\mathbb{C}_-$.  Otherwise,  $S(z)$  defined by \eqref{ups4} is a meromorphic matrix-valued function in $\mathbb{C}_-$ 
and it can be considered as a characteristic function of  ${\bf H}_{a{\bf q}}$.  
Lemmas \ref{nnn1}-\ref{AGH789} and Corollary \ref{nnn125} 
justify such a point of view by showing a close relationship between properties of non-self-adjoint ${\bf H}_{a{\bf q}}$
 and theirs $S$-matrices. Examples of $S$-matrices for various non-cyclic ${\bf q}$ are given in Sec. \ref{sec.5.1}. 

Throughout the paper,  $\mathcal{D}(H)$, $\mathcal{R}(H)$, and $\ker{H}$ denote the
domain, the range, and the null-space of a linear operator $H$, respectively,
whereas $H\upharpoonright_{\mathcal{D}}$ stands for the restriction of
$H$ to the set $\mathcal{D}$ and $\bigvee_{t\in\mathbb{R}}X_t$ means the closure of linear span of sets $X_t$. 
The symbol  $H^2(\mathbb{C}_+)$, where $\mathbb{C}_+=\{z\in\mathbb{C} : Im\ z >0 \}$  is used for  the Hardy space.  
The  Sobolev space is denoted as $W_2^p(I)$ ($I\in\{\mathbb{R}, \mathbb{R}_+\}$, $p\in\{1, 2\}$).
 
\section{Elements of Lax-Phillips scattering theory}\label{sec.2}
Here all necessary results about the Lax-Phillips scattering theory are presented.
The  monographs \cite{LF},  \cite[Chap. III]{KK} and the papers  \cite{KU1, KU2}  are recommended as complementary reading on the subject.
\subsection{Applicability of the Lax-Phillips scattering approach}
A  continuous group of unitary operators $W(t)$ acting in a Hilbert space $\mathfrak{W}$
is a subject of the Lax-Phillips scattering theory \cite{LF} 
if there exist so-called \emph{incoming} $D_-$ and \emph{outgoing} $D_+$ subspaces of $\mathfrak{W}$ with properties:
$$
\begin{array}{l}
(i) \quad W(t)D_+\subset{D_+}, \qquad W(-t)D_-\subset{D_-}, \quad t\geq{0}; \vspace{3mm} \\
(ii) \quad \bigcap_{t>0}W(t)D_+=\bigcap_{t>0}W(-t)D_-=\{0\}.
\end{array}
$$

Conditions $(i)-(ii)$ allow to construct incoming and outgoing 
spectral representations  for the restrictions of $W(t)$ onto
the subspaces
\begin{equation}\label{AGH21}
M_-={\bigvee_{t\in{\mathbb R}}W(t)D_{-}} \quad \mbox{and} \quad
M_+={\bigvee_{t\in{\mathbb R}}W(t)D_{+}},
\end{equation}
respectively  and define the corresponding Lax--Phillips scattering matrix ${S}(\delta)$ \
$(\delta\in{\mathbb R})$ whose values are contraction operators \cite{Arov},  \cite[Chap. 3]{KK}.
Furthermore, the additional  condition of orthogonality 
$$
 (iii) \quad  D_-\perp{D_+}
$$ 
guarantees that ${S}(\delta)$  is the boundary value of a contracting 
 operator-valued function ${S}(z)$ holomorphic in the lower half-plane $\mathbb{C}_-$ \cite[p. 52]{LF}.

Usually, the Lax-Phillips scattering matrix is defined with the use of an operator-differential equation
\begin{equation}\label{bonn60}
\frac{d^2}{dt^2}u=-Hu,
\end{equation}
where $H$ is a positive\footnote{i.e. $(Hf,f)>0$ for nonzero $f\in\mathcal{D}(H)$} 
self-adjoint operator in a Hilbert space $\mathfrak{H}$.
Denote by  ${\mathfrak H}_{H}$ the completion of 
$\mathcal{D}(H)$ with respect to the norm $\|{\cdot}\|_{H}^2:=(H{\cdot},\cdot)$.

The Cauchy problem for  \eqref{bonn60} determines a continuous group of unitary operators $W(t)$ 
in the space 
$$
\mathfrak{W}={\mathfrak H}_{H}\oplus{\mathfrak H}=\left\{\left[\begin{array}{c}
u \\
v 
\end{array}\right] \ : \   u\in\mathfrak{H}_H, \quad v\in\mathfrak{H}\right\}.
$$

If $H=-\Delta$ and $\mathfrak{H}=L_2(\mathbb{R}^n)$, then 
\eqref{bonn60} coincides with  the wave equation  $u_{tt}=\Delta{u}$ and
the corresponding subspaces $D_\pm$ constructed in  \cite{LF} 
possess the additional property  
\begin{equation}\label{e15}
 \quad JD_-=D_+,
\end{equation}
where $J$ is a self-adjoint and unitary operator in $\mathfrak{W}$ (so-called time-reversal operator):
\begin{equation}\label{AGH3}
J\left[\begin{array}{c}
u \\
v 
\end{array}\right]=\left[\begin{array}{c}
u \\
-v 
\end{array}\right].
\end{equation}
Relation \eqref{e15}  is a characteristic property of dynamics governed by  
 wave equations.

It is clear that, the existence of subspaces $D_\pm$ for $W(t)$  is determined by specific properties of $H$
in \eqref{bonn60}.  Before explaining which properties of $H$ are needed, we recall that
a symmetric operator $B$ is called \emph{simple}
 if its restriction on any nontrivial reducing subspace is not a
 self-adjoint operator. The maximality of  $B$ means that
 there are no symmetric extensions of $B$. The latter is equivalent to the fact that one of defect numbers of $B$ is equal to zero.
In what follows, without loss of generality, we assume that $B$ has  \emph{zero defect number} in $\mathbb{C}_+$, i.e.,
 $\dim\ker(B^*-{i}I)=0$, where $B^*$ is the adjoint of $B$. The latter means that
 \begin{equation}\label{GGG2}
  \ker({B^*}^2-\mu^2{I})=\ker({B^*}-\mu{I}), \qquad \mu\in\mathbb{C}_-.
 \end{equation} 
\begin{theorem}\cite{KU2, KK}\label{new31b}  	
Let $H$ be a positive self-adjoint operator  in a Hilbert space ${\mathfrak H}$. 	
The following are equivalent:
\begin{enumerate}
\item[(i)]  the group $W(t)$ of solutions of the Cauchy problem of \eqref{bonn60} has subspaces $D_\pm$
with properties $(i)-(iii)$ and \eqref{e15};
\item[(ii)] there exists a simple maximal symmetric operator $B$ acting in a subspace ${\mathfrak
H}_0$ of ${\mathfrak H}$ such that $H$ is an extension
(with exit in the space ${\mathfrak H}$) of the symmetric operator $B^2$.
\end{enumerate}
\end{theorem}

\subsection{The Lax-Phillips scattering matrix and its analytical continuation}\label{sec2.2} 
By Theorem \ref{new31b}, the unitary group $W(t)$
can be investigated by the Lax-Phillips scattering methods if and only if  $H$ is an extension of a symmetric operator $B^2$ acting in 
a subspace $\mathfrak{H}_0$ of $\mathfrak{H}$. 
A simple maximal symmetric operator $B$ in Theorem \ref{new31b} turns out to be a useful technical tool allowing one
to exhibit principal parts of the Lax-Phillips theory in a simple form. 
In particular, the subspaces $D_\pm$  coincide with the closure\footnote{in the space $\mathfrak{W}$} 
of the sets:
\begin{equation}\label{e2b}
\left\{\left[\begin{array}{c}
u \\ iBu
\end{array}\right] \ \left|\right. \ \forall{u}\in\mathcal{D}(B^2) \right\} \quad
\mbox{and} \quad
\left\{\left[\begin{array}{c}
u \\
-iBu
\end{array}\right] \ \left|\right. \ \forall{u}\in\mathcal{D}(B^2) \right\},
\end{equation}
respectively. Moreover, for all $t\geq{0}$,
\begin{equation}\label{e3b}
W(t)\left[\begin{array}{c}
u \\ iBu
\end{array}\right]=\left[\begin{array}{c}
V(t)u \\ iBV(t)u
\end{array}\right], \quad  W(-t)\left[\begin{array}{c}
u \\ -iBu
\end{array}\right]=\left[\begin{array}{c}
V(t)u \\ -iBV(t)u
\end{array}\right],  
\end{equation}
where $V(t)=e^{iBt}$ is a semigroup of isometric operators in $\mathfrak{H}_0$.

 The formulas \eqref{AGH21}, \eqref{e2b},  and \eqref{e3b}
allow one to construct the incoming/outgoing spectral representations for the restrictions of $W(t)$ onto $M_\pm$ 
in an explicit form \cite[Sec. 2.1]{Gawlik}. The latter leads to a simple method for the calculation of the
Lax-Phillips scattering matrix $S(\cdot)$ \cite{CK, AlAn}.  Actually, we need only a positive boundary triplet\footnote{see \cite[Chap 3]{GORB}
for definition of boundary triplets and positive boundary triplets}  $({\mathcal H}, \Gamma_{0},\Gamma_{1})$ of ${B^*}^2$ defined 
as follows: denote  ${\mathcal H}=\ker({B^*}^2+I)$, then $\mathcal{D}({B^*}^2)=\mathcal{D}(B^*B)\dot{+}{\mathcal H}$ and
each vector $f\in\mathcal{D}({B^*}^2)$  can be decomposed:
\begin{equation}\label{bonn41b}
f=u+h, \qquad u\in{\mathcal{D}({B^*}{B})}, \qquad h\in{\mathcal H}.
\end{equation}

The formula (\ref{bonn41b}) allows to define the linear mappings $\Gamma_{i} : \mathcal{D}({B^*}^2) \to {\mathcal H}$
 \begin{equation}\label{e7b}
 \Gamma_{0}f=\Gamma_{0}(u+h)=h,  \qquad \Gamma_{1}f=\Gamma_{1}(u+h)=P_{{\mathcal H}}(B^*B+I)u,
    \end{equation}
  where $P_{{\mathcal H}}$ is the orthogonal projector of ${\mathfrak H}_0$ onto the subspace ${\mathcal H}$.

\begin{theorem}[\cite{CK, AlAn}]\label{esse3}
If conditions of Theorem \ref{new31b} hold, then the Lax-Phillips scattering matrix 
$S(\cdot)$ for  the unitary group $W(t)$  of Cauchy problem solutions of \eqref{bonn60}
has the following analytical continuation into $\mathbb{C}_-$:
\begin{equation}\label{red1}
{S}(z)=[I-2(1+iz)C(z)][I-2(1-iz)C(z)]^{-1},  \qquad z\in\mathbb{C}_-,
\end{equation}
where the operators $C(z): {\mathcal H}\to{\mathcal H}$ are determined  by the relation
\begin{equation}\label{tttt}
C(z)\Gamma_{1}u=\Gamma_{0}u,  \qquad u{\in}P_{{\mathfrak H}_0}({H}-z^2{I})^{-1}\ker(B^*+\overline{z}I), \quad z\in\mathbb{C}_- .
\end{equation}
\end{theorem}
An investigation of $C(z)$ carried out in \cite{AlAn} shows that
the values of $S(z)$ are contraction operators in $\mathcal{H}$ and ${S}^*(z)=S(-\overline{z})$.

In what follows,  the analytical continuation \eqref{red1} of the  Lax-Phillips scattering matrix 
will be called the \emph{$S$-matrix} of the positive self-adjoint operator $H$ in \eqref{bonn60}. 
For this reason it is natural
to ask: \emph{to what extend the $S$-matrix determines  $H$?}

We recall that a self-adjoint operator $H$ is called  \emph{minimal} if 
each subspace of $\mathfrak{H}\ominus\mathfrak{H}_0$ that reduces $H$ is trivial.
Minimal self-adjoint extensions $H_1$  and $H_2$ of $B^2$ are called
\emph{unitary equivalent} if there exists an unitary operator $Z$ in $\mathfrak{H}$ 
such that $ZH_1=H_2Z$ and $Zf=f$ for all $f\in\mathfrak{H}_0$.

It follows from  \cite{AlAn} that the ${S}$-matrix  determines 
a minimal positive self-adjoint extension $H$ of $B^2$ up to unitary equivalence.

\begin{remark}\label{BBB2}
Various approaches in non-stationary scattering theory are based on  the comparing of two evolutions:  ``unperturbed''  and ``perturbed''.
The subspaces $D_\pm$ characterize  unperturbed evolution in the Lax-Phillips approach.
Due to \eqref{e2b},  the subspaces $D_\pm$  are described by the operator $B$. The operator $B^*B$ is a positive self-adjoint 
extension of $B^2$ in the space $\mathfrak{H}_0$ and the group $W_0(t)$ of solutions of the Cauchy problem of \eqref{bonn60}
(with $B^*B$ instead of $H$) determines an unperturbed evolution. The corresponding wave operators 
$\Omega_\pm=s-\lim_{t\to\pm\infty}W(-t)W_0(t)$ exist and are isometric in $\mathfrak{H}_0$. The scattering
operator $\Omega_+^{*}\Omega_-$ coincides with the Lax-Phillips scattering matrix $S(\delta)$ in the spectral
representation of the unperturbed evolution $W_0(t)$ \cite{AlAn}.
\end{remark}  

\section{Properties of operators ${\bf H}_{a{\bf q}}$}\label{sec.3}
\subsection{Preliminaries}\label{sec3.1}
For technical reasons it is convenient to calculate the $S$-matrix  
for unitary equivalent copy of the operator $H_{aq}$ 
 in the Hilbert space  $L_2(\mathbb{R}_+, \mathbb{C}^2)$.
To do that, for each function ${f}\in{L_2(\mathbb{R})}$, we define the operator\footnote{we will use the mathbf font for $\mathbb{C}^2$-valued 
functions of $L_2(\mathbb{R}_+, \mathbb{C}^2)$ in order to avoid confusion with functions from ${L}_2(\mathbb{R})$. In particular,
${\bf e}^{-i{\mu}x}\equiv\left[\begin{array}{c}
 e^{-i{\mu}x}  \\
e^{-i{\mu}x} 
\end{array}\right]$.} 
$$
Yf=\left[\begin{array}{c}
f(x)   \\
f(-x) 
\end{array}\right]={\bf f}(x), \qquad  x>{0}
$$
that maps isometrically $L_2(\mathbb{R})$ onto $L_2(\mathbb{R}_+, \mathbb{C}^2)$ 
and maps ${W}_2^2(\mathbb{R}\backslash\{0\})$ onto  ${W}_2^2(\mathbb{R}_+, \mathbb{C}^2)$.
For all  ${\bf f}=Yf$, $f\in{W}_2^2(\mathbb{R}\backslash\{0\})$  we denote
$[{\bf f}]_r=f_r(0)$ and  $[{\bf f}]_s=f_s(0)$. In other words,
\begin{equation}\label{GGG14}
[{\bf f}]_r=\frac{1}{2}\lim_{x\to+0}(f_1(x)+f_2(x)), \quad  [{\bf f}]_s=\lim_{x\to+0}(f_1(x)-f_2(x)), \quad {\bf f}=\left[\begin{array}{c}
f_1  \\
f_2 
\end{array}\right].
\end{equation}

It is easy to see that    $YH_{aq}={\bf H}_{a\sf{q}}{Y}$, 
where $H_{aq}$ is defined by \eqref{AGH1}, \eqref{ggg1} and the operator
\begin{equation}\label{newww6a}
{\bf H}_{a{\bf q}}{\bf f}=-\frac{d^2{\bf f}}{dx^2}+[{\bf f}]_r{\bf q}(x),  \qquad  {\bf q}=\left[\begin{array}{c}
q_1  \\
q_2
\end{array}\right]=Yq
\end{equation}
acts in $L_2(\mathbb{R}_+, \mathbb{C}^2)$ with domain of definition 
\begin{equation}\label{newww6}
\mathcal{D}({\bf H}_{a{\bf q}})=\{{\bf f}\in W_2^2(\mathbb{R}_+,\mathbb{C}^2): \
 [{\bf f}]_s=0,  \quad  [{\bf f'}]_r= {a}[{\bf f}]_r+({\bf f}, {\bf q})_{+} \},
\end{equation}
where $({\bf f}, {\bf q})_{+}=(Yf,Yq)_+=(f,q)$ is the scalar product in $L_2(\mathbb{R}_+, \mathbb{C}^2)$.

When $a\to\infty$, the formulas  \eqref{newww6a} and \eqref{newww6} determine  a positive self-adjoint operator 
in $L_2(\mathbb{R}_+, \mathbb{C}^2)$
$$
{\bf H}_{\infty}\equiv{\bf H}_{\infty{\bf q}}=-\frac{d^2}{dx^2}, \quad \mathcal{D}({\bf H}_{\infty})=\{{\bf f}\in W_2^2(\mathbb{R}_+,\mathbb{C}^2): \ {\bf f}(0)=0\}
$$
that does not depend on the choice of ${\bf q}$ and can be decomposed 
$$
{\bf H}_{\infty}{\bf f}=\left[\begin{array}{c}
H_\infty{f}_1  \\
H_\infty{f}_2 
\end{array}\right], \quad H_\infty=-\frac{d^2}{dx^2}, \quad \mathcal{D}({H}_{\infty})=\{{f}\in W_2^2(\mathbb{R}_+): \ {f}(0)=0\}.
$$

By analogy with \cite[Sec. 5]{KZ} (where the case of operators $H_{aq}$ has been studied)
we consider ${\bf H}_{a{\bf q}}$ and  ${\bf H}_\infty$  as restrictions of the maximal operator
$$
{\bf H}_{max}{\bf f}=-\frac{d^2{\bf f}}{dx^2}+[{\bf f}]_r{\bf q}(x), \qquad \mathcal{D}({\bf H}_{max})=\{{\bf f}\in W_2^2(\mathbb{R}_+,\mathbb{C}^2): \
[{\bf f}]_s=0 \}.
$$
onto the corresponding domain of definition. 

The maximal operator ${\bf H}_{max}$ has a boundary triplet $(\mathbb{C}, \Gamma_0, \Gamma_1)$,  where 
\begin{equation}\label{DDD5}
\Gamma_0{\bf f}=[{\bf f}]_r, \qquad \Gamma_1{\bf f}=2[{\bf f'}]_r-({\bf f}, {\bf q})_+,  \quad {\bf f}\in\mathcal{D}({\bf H}_{max})
\end{equation}
and the formulas \eqref{newww6a} and \eqref{newww6} are rewritten:
\begin{equation}\label{DDD1}
{\bf H}_{a{\bf q}}={\bf H}_{max}\upharpoonright_{\mathcal{D}({\bf H}_{a{\bf q}})}, \quad \mathcal{D}({\bf H}_{a{\bf q}})=\{{\bf f}\in\mathcal{D}({\bf H}_{max}) \ : \ 
a\Gamma_0{\bf f}=\Gamma_1{\bf f}\}.
\end{equation}
In particular,  ${\bf H}_\infty$ is the restriction of ${\bf H}_{max}$ onto $\ker\Gamma_0$ and
its resolvent is
\begin{equation}\label{newww2}
({\bf H}_\infty-z^2I)^{-1}{\bf f}=\frac{i}{2z}[{\bf A}_z(x)e^{-izx}+{\bf B}_z(x)e^{izx}], \qquad  {\bf f}\in{L_2(\mathbb{R}_+, \mathbb{C}^2)},   
\end{equation}
where $z\in\mathbb{C}_-$ and
$$
{\bf A}_z(x)=\int_0^\infty{e^{-izs}}{\bf f}(s)ds-\int_0^x{e^{izs}}{\bf f}(s)ds,  \quad {\bf B}_z(x)=-\int_x^\infty{e^{-izs}}{\bf f}(s)ds.
$$
\begin{lemma}\label{neww71}
The Krein-Naimark resolvent formula
\begin{equation}\label{DDD4}
({\bf H}_{a{\bf q}}-z^2{I})^{-1}{\bf f}=({\bf H}_\infty-z^2{I})^{-1}{\bf f}+\frac{({\bf f}, {\bf u}_{-\overline{z}})_{+}}{a-W(z^2)}{\bf u}_z(x)
\end{equation}
holds for $a\not=W(z^2)$. Here,
\begin{equation}\label{DDD}
{\bf u}_{\mu}(x)={\bf e}^{-i{\mu}x}-({\bf H}_\infty-\mu^2I)^{-1}{\bf q}, \qquad \ \mu\in\{z, -\overline{z}\}\subset\mathbb{C}_-
\end{equation}
is an eigenfunction of ${\bf H}_{max}$ corresponding to the eigenvalue $\mu^2$   and
\begin{equation}\label{DDD3}
W(z^2)=-2iz-2({\bf e}^{-izx}, Re \ {\bf q})_{+}+(({\bf H}_\infty-z^2I)^{-1}{\bf q}, {\bf q})_{+},
 \quad z\in\mathbb{C}_-.
\end{equation}
\end{lemma}
\begin{proof} 
It follows from \cite{KZ} that the subspace $\ker({\bf H}_{max}-\mu^2I)$  is one dimensional and it is generated by the function 
${\bf u}_{\mu}$ defined by \eqref{DDD}.  Setting $\mu=z$ and  using \eqref{DDD5}, we conclude that 
$\Gamma_0{\bf u}_{z}=1$ and the Weyl-Titchmarsh function associated to  the boundary triplet  $(\mathbb{C}, \Gamma_0,
\Gamma_1)$ takes the form  
$$
W(z^2)=\Gamma_1{\bf u}_z=-2iz-2[{\bf v'}]_r-({\bf e}^{-izx}+{\bf v}, {\bf q})_+,
$$
where ${\bf v}=({\bf H}_\infty-z^2I)^{-1}{\bf q}$.  In view of \eqref{newww2},
${\bf v'}(0)=\int_0^\infty{e^{-izs}}{\bf q}(s)ds$ and hence,
$$
2[{\bf v'}]_r+({\bf e}^{-izx}, {\bf q})_+=2({\bf e}^{-izx}, Re \ {\bf q})_+, \quad Re \ {\bf q}=\left[\begin{array}{c}
Re \ q_1  \\
Re \ q_2 
\end{array}\right].
$$
Substituting this  expression into the formula for $W(z^2)$ we obtain \eqref{DDD3}.

In terms of the boundary triplet $(\mathbb{C}, \Gamma_0, \Gamma_1)$, the Krein-Naimark resolvent formula has the form 
\cite[Theorem 14.18, Proposition 14.14]{Schm}
$$
({\bf H}_{a{\bf q}}-z^2{I})^{-1}{\bf f}=({\bf H}_\infty-z^2{I})^{-1}{\bf f}+\frac{\Gamma_1{\bf u}}{a-W(z^2)}{\bf u}_z(x), 
$$
where ${\bf u}=({\bf H}_\infty-z^2{I})^{-1}{\bf f}.$  
In view of \eqref{newww2},  ${\bf u'}(0)=\int_0^\infty{e^{-izs}}{\bf f}(s)ds$. Taking \eqref{GGG14} into account,
$$
2[{\bf u'}]_r=\int_0^\infty{e^{-izs}}({f_1(s)+f_2(s)})dx=({\bf f},  {\bf e}^{i\overline{z}x})_+.
$$
Finally,  using  \eqref{DDD5} and \eqref{DDD} with $\mu=-\overline{z}$, we obtain
$$
\Gamma_1{\bf u}=({\bf f},  {\bf e}^{i\overline{z}x})_{+}-({\bf u}, {\bf q})_{+}=({\bf f},  {\bf e}^{i\overline{z}x}-({\bf H}_\infty-\overline{z}^2{I})^{-1} {\bf q})_{+}= 
({\bf f},  {\bf u}_{-\overline{z}})_{+}
$$
that completes the proof.
\end{proof}

 \subsection{Applicability of the Lax-Phillips approach for ${\bf H}_{a{\bf q}}$}
 Denote by
  \begin{equation}\label{AGH121bc}
  \mathcal{B}=i\frac{d}{dx}, \qquad \mathcal{D}(\mathcal{B})=\{u\in{{W}_{2}^{1}}({\mathbb R}_+) :  u(0)=0\}
  \end{equation}
 the first derivative operator in $L_2(\mathbb{R}_+)$. The \emph{same notation} will be used for its analog  acting in   
 $L_2(\mathbb{R}_+, \mathbb{C}^2)$.  The both operators
 are simple maximal symmetric with zero defect numbers in $\mathbb{C}_+$,  and theirs Cayley transforms
 \begin{equation}\label{KKKK1}
T=({\mathcal B}-iI)({\mathcal B}+iI)^{-1}
\end{equation}
are forward shift operators  in the corresponding spaces.
 
A function ${\bf q}\in L_2(\mathbb{R}_+, \mathbb{C}^2)$ is called \textit{non-cyclic} for the backward shift operator $T^*$ if the subspace
$$
E_{\bf q}=\bigvee _{n=0}^{\infty}{T^*}^n{\bf q}
$$
\emph{does not coincide} with  $L_2(\mathbb{R}_+, \mathbb{C}^2)$. 

Considering $L_2(\mathbb{R}_+)$ as a subspace of $L_2(\mathbb{R})$ we conclude that
the Fourier transform
$$
Ff(\delta)=\frac{1}{\sqrt{2\pi}}\int_{-\infty}^{\infty}e^{i\delta s}f(s)ds
$$
maps isometrically $L_2(\mathbb{R}_+)$ onto the Hardy space $H^2(\mathbb{C}_+)$ and
$$
F{\mathcal B}u=\delta Fu, \quad F{T}f=\frac{\delta-i}{\delta+i}Ff, \qquad u\in\mathcal{D}({\mathcal B}), \quad  f\in{L_2(\mathbb{R}_+)}. 
$$

Let $\psi\in H^{\infty}(\mathbb C_+)$ be an inner function.
Then  
\begin{equation}\label{AGHH}
\psi ({\mathcal B})=F^{-1}\psi(\delta)F
\end{equation}
 is an isometric operator in $L_2(\mathbb{R}_+)$ which commutes with ${\mathcal B}$ \cite[Sec. 5]{Gawlik}.

\begin{lemma}\label{zzzz}
The following are equivalent:
\begin{enumerate}
\item[(i)]  a function ${\bf q}=\left[\begin{array}{c}
q_1  \\
q_2 
\end{array}\right]$  is non-cyclic for the backward shift operator $T^*$; 
\item[(ii)] there exists an inner function $\psi\in H^{\infty}(\mathbb C_+)$ such that
the subspace ${\mathfrak H}_0=\psi ({\mathcal B})L_2(\mathbb{R}_+)$ of $L_2(\mathbb{R}_+)$ is orthogonal to 
at least one of the functions $q_i$.
\end{enumerate}
\end{lemma}
\begin{proof} 
$(i)\to(ii)$ Since $E_{\bf q}=E_{q_1}{\oplus}E_{q_2}$, 
the function ${\bf q}$ is non-cyclic if and only if at least one of the functions $q_i{\in}L_2(\mathbb{R}_+)$
is non-cyclic for the  backward shift operator $T^*$ in $L_2(\mathbb{R}_+)$.
Let $q\equiv{q}_i$ be non-cyclic. Then the non-zero subspace 
$$
{\mathfrak H}_0=L_2(\mathbb{R}_+)\ominus E_{q}
$$
is invariant with respect to $T$. 
This means that  $F{\mathfrak H}_0$ is invariant  
with respect to the multiplication by $\frac{\delta-i}{\delta+i}$ in $H^2(\mathbb{C}_+)$. 
The Beurling theorem \cite[p. 164]{LB} yields the existence of an inner function $\psi \in H^\infty(\mathbb{C}_+)$ such that 
$F{\mathfrak H}_0=\psi(\delta)H_2(\mathbb{C}_+)$. Therefore 
 $$
 {\mathfrak H}_0=F^{-1}\psi(\delta)FL_2(\mathbb{R}_+)=\psi({\mathcal B})L_2(\mathbb{R}_+).
 $$
 By the construction, ${\mathfrak H}_{0}$ is orthogonal to $q$ (since, $q$ belongs to $E_{q}$).
  
  $(ii)\to(i)$  Let  ${\mathfrak H}_{0}=\psi ({\mathcal B})L_2(\mathbb{R}_+)$ be orthogonal to $q$. Then\footnote{here, $(\cdot,\cdot)_+$ is the scalar product in
  $L_2(\mathbb{R}_+)$.}  
  $$
  (\psi({\mathcal B})f, {{T}^*}^nq)_+=({{T}}^n\psi({\mathcal B})f, q)_+=(\psi({\mathcal B}){T}^nf, q)_+=0 \quad \mbox{for all} \quad f\in{L_2(\mathbb{R}_+)}.
  $$
  Therefore, ${{T}^*}^nq$ is orthogonal to ${\mathfrak H}_0$.
   This means that $E_{q}$   is orthogonal to ${\mathfrak H}_0$. Therefore,
  $E_{q}$ is a proper subspace of $L_2(\mathbb R_+)$ and $q$ is non-cyclic.
  \end{proof}

\begin{theorem}\label{yyyy}
If $\bf q$  is non-cyclic for $T^*$, then there exists a simple maximal symmetric operator $B$ acting in a subspace  ${\mathfrak H}_{0}$
of $L_2(\mathbb{R}_+, \mathbb{C}^2)$ such that the operators ${\bf H}_{a{\bf q}}$ 
are extensions of the symmetric operator $B^2$ for all ${a}\in\mathbb{C}$.
\end{theorem}
\begin{proof}  
If $\bf q$  is non-cyclic, then at least one of $q_i$ is non-cyclic.  Consider firstly the case where the both of functions $q_i$ are 
non-cyclic. Due to the proof of  Lemma \ref{zzzz},  for each $q_i$ there exists an inner function $\psi_i$ such that the subspace
 $\psi_i({\mathcal B})L_2(\mathbb{R}_+)$ is orthogonal to $q_i$. Denote
 \begin{equation}\label{AGH33b}
 {\mathfrak H}_{0}=\left[\begin{array}{l} 
 \psi_1({\mathcal B})L_2(\mathbb{R}_+) \\
 \psi_2({\mathcal B})L_2(\mathbb{R}_+)
 \end{array}\right]=\psi({\mathcal B})L_2(\mathbb{R}_+, \mathbb{C}^2), 
\end{equation}
where  
\begin{equation}\label{AGH77} 
\psi({\mathcal B})=\left[\begin{array}{cc}
\psi_1({\mathcal B})  & 0 \\
0 & \psi_2({\mathcal B})
\end{array}\right]
\end{equation}
is an isometric operator in $L_2(\mathbb{R}_+, \mathbb{C}^2)$ that commutes with ${\mathcal B}$. This allows 
to define a simple maximal symmetric operator in $\mathfrak{H}_0$:
\begin{equation}\label{e7}
B=\psi({\mathcal B}){\mathcal B}\psi({\mathcal B})^*, \qquad  \mathcal{D}(B)=\psi({\mathcal B})\mathcal{D}({\mathcal{B}}).
\end{equation}

Since $\psi({\mathcal B})$ commutes with ${\mathcal B}$, the formula \eqref{e7}
can be rewritten as
\begin{equation}\label{pppp}
B{\bf u}=\mathcal{B}{\bf u}, \qquad {\bf u}\in\mathcal D(B)=\psi({\mathcal B})\mathcal{D}({\mathcal{B}})=
\mathcal{D}({\mathcal{B}})\cap{\mathfrak H}_0.
\end{equation}
 (i.e., $B$ is a part of ${\mathcal B}$ restricted on $\mathfrak{H}_{0}$).
 In view of \eqref{AGH121bc} and \eqref{pppp}
\begin{equation}\label{dddd}
B^2=-\frac{d^2}{dx^2}, \quad 
\mathcal D(B^2)=\{{\bf u}\in{{W}_{2}^{2}}({\mathbb R}_+, \mathbb{C}^2)\cap{\mathfrak H}_0 :  {\bf u}(0)={\bf u}'(0)=0\}.
\end{equation}
 
By  Lemma \ref{zzzz} and \eqref{AGH33b}, the subspace  ${\mathfrak H}_{0}$  is orthogonal to ${\bf q}$.
Hence, in view of \eqref{newww6a},  \eqref{newww6}, and \eqref{dddd},
$\mathcal D({\bf H}_{a{\bf q}})\supset\mathcal D({B}^2)$ and 
$$
{\bf H}_{a{\bf q}}{\bf u}=-\frac{d^2\bf u}{dx^2}=B^2{\bf u} \quad \mbox{for all} \quad {\bf u}\in\mathcal D(B^2).
$$

The case where only one $q_i$ is considered similarly. For example, if $q_1$ is non-cyclic whereas $q_2$ is cyclic (i.e., $E_{q_2}=L_2(\mathbb{R}_+)$), then
${\mathfrak H}_0$ and $\psi({\mathcal B})$ are determined as above with  $\psi_2=0$.
\end{proof}
\begin{corollary}\label{ffff} 
Assume  that $H={\bf H}_{a{\bf q}}$ is a positive self-adjoint operator. If  ${\bf q}$ is non-cyclic for $T^*$, then
the group $W(t)$ of Cauchy problem solutions of \eqref{bonn60} has incoming/outgoing subspaces $D_\pm$
defined by \eqref{e2b}, where $B$ is from \eqref{pppp}.
\end{corollary}
\begin{proof}
It follows from Theorems \ref{new31b} and \ref{yyyy}.
\end{proof}
 
\section{$S$-matrix for positive self-adjoint operator}\label{sec.4}

In this section we suppose that ${\bf H}_{a{\bf q}}$    
is a positive self-adjoint operator and the function ${\bf q}$ is non-cyclic.
By Theorem \ref{yyyy},  ${\bf H}_{a{\bf q}}$ is an extension of the symmetric operator  $B^2$ defined by \eqref{dddd} that acts in the subspace
${\mathfrak H}_0=\psi({\mathcal B})L_2(\mathbb{R}_+, \mathbb{C}^2)$. 
In view of Corollary \ref{ffff} and Theorem \ref{esse3},  the $S$-matrix
of ${\bf H}_{a{\bf q}}$  exists and is given by \eqref{red1}. Our goal is to modify this general formula 
taking into account the specific choice of $B$ in \eqref{pppp}. 

\subsection{Preliminaries}
The following technical results are needed for the calculation of $S$-matrix.
\begin{lemma}\label{AGH44}
Let an isometric operator $\psi(\mathcal{B})$ be defined by \eqref{AGHH}. Then
$$
\psi(\mathcal{B})^*{e}^{-i{\mu}x}=\overline{\psi(\overline{\mu})}e^{-i{\mu}x}, \qquad \mu\in\mathbb{C}_-.
$$
\end{lemma}
\begin{proof} 
It follows from \eqref{AGH121bc} that  
$\mathcal{B}^*=i\frac{d}{dx}, \ \mathcal{D}(\mathcal{B}^*)={{W}_{2}^{1}}({\mathbb R}_+)$. Therefore, \
$\ker({{\mathcal B}^*}-\mu{I})=\{c{e}^{-i{\mu}x} : c\in\mathbb{C}\}$. This means that, for all $u\in\mathcal{D}({\mathcal B})$, 
$$
(({\mathcal B}-\overline{\mu}I)u, \psi(\mathcal{B})^*e^{-i{\mu}x})_+=(\psi(\mathcal{B})({\mathcal B}-\overline{\mu}I)u, e^{-i{\mu}x})_+=(({\mathcal B}-\overline{\mu}I)\psi(\mathcal{B})u, e^{-i{\mu}x})_+=0.
$$
Hence $\psi(\mathcal{B})^*e^{-i{\mu}x}$ belongs to $\ker({{\mathcal B}^*}-\mu{I})$ and
\begin{equation}\label{bbb1}
(\psi(\mathcal{B})^*e^{-i{\mu}x}, e^{-i{\mu}x})_+=c(e^{-i{\mu}x}, e^{-i{\mu}x})_+=-\frac{c}{2Im\ \mu}.
\end{equation}
Using \eqref{AGHH} and taking into account that  $F\chi_{{\mathbb R}_+}(x)e^{-i{\mu}x}=\frac{i}{\sqrt{2\pi}}\cdot\frac{1}{\delta-\mu}$, 
we verify that the inner product 
$$
(\psi(\mathcal{B})^*e^{-i{\mu}x}, e^{-i{\mu}x})_+=(e^{-i{\mu}x}, \psi(\mathcal{B})e^{-i{\mu}x})_+=(F\chi_{{\mathbb R}_+}(x)e^{-i{\mu}x}, \psi(\delta)F\chi_{{\mathbb R}_+}(x)e^{-i{\mu}x})
$$ 
is equal to  ${\frac{1}{2\pi}\int_{-\infty}^\infty\frac{\overline{\psi(\delta)}}{(Re \ \mu-\delta)^2+(Im\ \mu)^2}d\delta}$.
The Poisson formula \cite[p.147]{Nik} and \eqref{bbb1} lead to the conclusion that
$$
c={\frac{1}{\pi}\int_{-\infty}^\infty\frac{-(Im \ \mu)\overline{\psi(\delta)}}{(Re \ \mu-\delta)^2+(Im\ \mu)^2}d\delta}=\overline{\psi(Re\ \mu-iIm\ \mu)}=\overline{\psi(\overline{\mu})}
$$
that completes the proof.
 \end{proof}
 
\begin{lemma}\label{GGG24} 
 Let  $B$ and $\psi(\mathcal{B})$ be defined by \eqref{e7} and \eqref{AGH77}, respectively. Then, for any $\mu\in\mathbb{C}_-$,
$$
\ker({B^*}^2-\mu^2I)=\ker({B}^*-\mu{I})= \psi({\mathcal B})\left\{ {\bf h}_{\mu}=\left[\begin{array}{c}
{\alpha}_\mu  \\
{\beta}_\mu
\end{array}\right]e^{-i{\mu}x}  \ : \  \alpha_\mu, \beta_\mu\in\mathbb{C} \right\}.
$$
\end{lemma} 
 \begin{proof} The first identity follows from \eqref{GGG2}.
It follows from \eqref{e7} that 
 \begin{equation}\label{AGH23}
 B^*=\psi({\mathcal B}){\mathcal B}^*\psi({\mathcal B})^*,  \quad \mathcal{D}(B^*)=\psi({\mathcal B})\mathcal{D}({\mathcal B}^*)=\psi({\mathcal B})W_2^1({\mathbb R}_+, \mathbb{C}^2).   
 \end{equation}
By virtue of \eqref{AGH23} we conclude that  $\ker({B}^*-\mu{I})=\psi({\mathcal B})\ker(\mathcal{B}^*-\mu{I})$.
It follows from the proof of Lemma \ref{AGH44} that  $\ker(\mathcal{B}^*-\mu{I})$ coincides with the set
of vectors $\{{\bf h}_{\mu}\}$ defined above.
 \end{proof}
 
 \begin{corollary}\label{GGG15}
Let $\psi(\mathcal{B})$ be defined by \eqref{AGH77}. Then,  for any $\mu\in\mathbb{C}_-$,
\begin{equation}\label{AGH79}
\psi(\mathcal{B})^*{\bf e}^{-i{\mu}x}=
\overline{\left[\begin{array}{c}
 {\psi_1(\overline{\mu})} \\
{\psi_2(\overline{\mu})}
\end{array}\right]}e^{-i{\mu}x}, \quad 
\psi(\mathcal{B})^*{\bf u}_{\mu}=\left[\begin{array}{c}
c(\mu, q_1) \\
c_(\mu, q_2) 
\end{array}\right]e^{-i{\mu}x},
\end{equation}
where ${\bf u}_{\mu}$ is defined by \eqref{DDD}  and
\begin{equation}\label{AGH44b}
c(\mu, q_j)=\overline{{\psi_j(\overline{\mu})}}+2(Im \ \mu)(({H}_\infty-\mu^2I)^{-1}q_j, \psi_j(\mathcal{B})e^{-i\mu{x}})_+.
\end{equation}
\end{corollary}
\begin{proof}
The first relation in \eqref{AGH79} follows from  Lemma \ref{AGH44}. 

The function ${\bf u}_{\mu}$ in the second relation is an eigenfunction of the operator 
${\bf H}_{max}$ (see Lemma \ref{neww71}). Since $(\mathbb{C}, \Gamma_0, \Gamma_1)$ defined by \eqref{DDD5} is a boundary triplet of ${\bf H}_{max}$,
its adjoint ${\bf H}_{max}^*$ coincides with the symmetric operator 
${\bf H}_{min}={\bf H}_{max}\upharpoonright_{\ker\Gamma_0\cap\ker\Gamma_1}$.
Precisely, 
$$
{\bf H}_{min}=-\frac{d^2}{dx^2},  \quad {\mathcal{D}({\bf H}_{min})}=\{{\bf f}\in W_2^2(\mathbb{R}_+,\mathbb{C}^2): \ [{\bf f}]_r=0, \ 2[{\bf f'}]_r=({\bf f}, {\bf q})_+\}.
$$
Comparing this formula with \eqref{dddd} leads to the conclusion that ${\bf H}_{min}\supset{B^2}$, i.e.,  ${\bf H}_{min}$ is an extension of $B^2$ with the exit into the 
 space $L_2(\mathbb{R}_+, \mathbb{C}^2)$.  
 Then, for ${\bf f}\in\mathcal{D}({\bf H}_{max})$ and ${\bf u}\in\mathcal{D}(B^2)$,
$$
(P_{{\mathfrak H}_0}{\bf H}_{max}{\bf f}, {\bf u})_+=({\bf H}_{max}{\bf f}, {\bf u})_+=({\bf f}, {\bf H}_{min}{\bf u})_+=(P_{{\mathfrak H}_0}{\bf f}, B^2{\bf u})_{+}=
({{B^*}^2}P_{{\mathfrak H}_0}{\bf f}, {\bf u})_{+},
$$
where $P_{{\mathfrak H}_0}$ is the orthogonal projection in $L_2(\mathbb{R}_+, \mathbb{C}^2)$ on the subspace 
${{\mathfrak H}_0}$ defined by \eqref{AGH33b}.
 The obtained relation means that 
 \begin{equation}\label{GGG55}
P_{{\mathfrak H}_0}{\bf H}_{max}{\bf f}={{B^*}^2}P_{{\mathfrak H}_{0}}{\bf f}, \quad \mbox{for all} \quad {\bf f}\in\mathcal{D}({\bf H}_{max})=W_2^2({\mathbb R}_+, \mathbb{C}^2).
\end{equation}
Setting ${\bf f}={\bf u}_\mu$ in \eqref{GGG55} and taking into account that  ${\bf H}_{max}{\bf u}_\mu=\mu^2{\bf u}_\mu$, we obtain
$P_{{\mathfrak H}_0}{\bf H}_{max}{\bf u}_\mu={{B^*}^2}P_{{\mathfrak H}_0}{\bf u}_\mu={\mu}^2P_{{\mathfrak H}_0}{\bf u}_\mu$.
This relation and \eqref{GGG2} mean
$$
P_{{\mathfrak H}_0}{\bf u}_\mu\in\ker({B^*}^2-\mu^2I)=\ker({B^*}-\mu{I}).
$$  
In view of Lemma \ref{GGG24}, $P_{{\mathfrak H}_0}{\bf u}_\mu=\psi(\mathcal{B}){\bf h}_\mu$ for some choice of
${\bf h}_\mu=\left[\begin{array}{c}
{\alpha}_\mu  \\
{\beta}_\mu
\end{array}\right]e^{-i{\mu}x}$ or $\psi(\mathcal{B})\psi(\mathcal{B})^*{\bf u}_\mu=\psi(\mathcal{B}){\bf h}_\mu$  since  
$P_{{\mathfrak H}_0}=\psi(\mathcal{B})\psi(\mathcal{B})^*$. Therefore
 $\psi(\mathcal{B})^*{\bf u}_\mu={\bf h}_\mu$ that leads to the second relation in \eqref{AGH79} with unspecified parameters $\alpha_\mu$, $\beta_\mu$.
Taking  \eqref{DDD}  into account  and arguing by the analogy with the determination of $c$ in the proof of Lemma \ref{AGH44} 
we arrive at the conclusion that  $\alpha_\mu=c(\mu, q_1)$ and  $\beta_\mu=c(\mu, q_2)$,  where 
$c(\mu, q_i)$ are defined in \eqref{AGH44b}.
\end{proof}
 
\subsection{Positive boundary triplet} In view of Sec. \ref{sec2.2}, the $S$-matrix
can not be constructed  without finding the positive boundary triplet $({\mathcal H}, \Gamma_{0},\Gamma_{1})$
 of ${B^*}^2$. Since $B$ is the restriction of the first derivative operator  $\mathcal{B}$ on $\mathfrak{H}_0$, see \eqref{pppp},  one can try to
 express $({\mathcal H}, \Gamma_{0},\Gamma_{1})$  in terms of well-known positive boundary triplet  $({\mathcal H}', \Gamma_{0}',\Gamma_{1}')$ 
 of ${\mathcal{B}^*}^2$.
\begin{lemma}\label{AGH678}
 The following relations hold:
$$
{\mathcal H}=\psi({\mathcal B}){\mathcal H}', \quad \Gamma_0\psi({\mathcal B})=\psi({\mathcal B})\Gamma_0'  \quad  \Gamma_1\psi({\mathcal B})=\psi({\mathcal B})\Gamma_1'. 
$$
\end{lemma}
\begin{proof} 
 It follows from \eqref{AGH23} that 
  \begin{equation}\label{AGH23b}
 {B^*}^2=\psi({\mathcal B}){{\mathcal B}^*}^2\psi({\mathcal B})^*,  \quad \mathcal{D}({B^*}^2)=\psi({\mathcal B})\mathcal{D}({{\mathcal B}^*}^2)=\psi({\mathcal B})W_2^2({\mathbb R}_+, \mathbb{C}^2)   
 \end{equation}
By definition ${\mathcal H}=\ker({B^*}^2+I)$ and ${\mathcal H}'=\ker({{\mathcal B}^*}^2+I)$. Using  \eqref{AGH23b}, we obtain
$$
{\mathcal H}=\ker({B^*}^2+I)=\psi({\mathcal B})\ker({{\mathcal B}^*}^2+I)=\psi({\mathcal B}){\mathcal H}'.
$$
It follows from \eqref{e7} and \eqref{AGH23} that
\begin{equation}\label{GGG1}
B^*B=\psi({\mathcal B}){\mathcal B}^*{\mathcal B}\psi({\mathcal B})^*, \qquad \mathcal{D}(B^*B)=\psi({\mathcal B})\mathcal{D}({\mathcal B}^*{\mathcal B})
\end{equation}

For brevity, we denote $V=\psi({\mathcal B})$  and consider ${\bf f}\in \mathcal{D}({{\mathcal B}^*}^2)$. 
Then ${\bf f}={\bf u}+{\bf h}$,  where ${\bf u}\in\mathcal {D}({\mathcal B}^*{\mathcal B})$ 
and ${\bf h}\in\mathcal{ H}'$.  By virtue of \eqref{AGH23b}, \eqref{GGG1}, $V{\bf f}\in\mathcal{D}({B^*}^2)$  and $V{\bf f}=V{\bf u}+V{\bf h}$, where
 $V{\bf u}\in \mathcal{D}(B^*B)$ and $V{\bf h}\in{\mathcal H}$.  In view of \eqref{e7b},
 $\Gamma_0V{\bf f}=V{\bf h}=V\Gamma_0'{\bf f}.$

Since ${\mathcal H}=V{\mathcal H}'$ and $\mathcal{R}(B^2+{I})=V\mathcal{R}({\mathcal B}^2+{I})$, 
the orthogonal projectors $P_{\mathcal H}$ and $P_{{\mathcal H}'}$ are related as follows: $VP_{{\mathcal H}'}=P_{\mathcal H}V$. 
Therefore, 
$$
\Gamma_1V{\bf f}=P_{\mathcal H}(B^*B+I)V{\bf u}=P_{\mathcal H}(V{\mathcal B}^*{\mathcal B}V^*+I)V{\bf u}=P_{\mathcal H}V({\mathcal B}^*{\mathcal B}+I){\bf u}=V\Gamma_1'{\bf f}
$$
that completes the proof.
\end{proof}

\begin{corollary}\label{kkkk}
The positive boundary triplet $({\mathcal H}, \Gamma_{0},\Gamma_{1})$ of ${B^*}^2$ consists of the space 
$$
{\mathcal H}=\psi({\mathcal B})\left\{\left[\begin{array}{c}
{\alpha}  \\
{\beta}
\end{array}\right]e^{-x} \ : \  \alpha, \beta\in\mathbb{C} \right\}
$$
and the mappings $\Gamma_{i} : \psi({\mathcal B})W_2^2({\mathbb R}_+, \mathbb{C}^2) \to {\mathcal H}$ that are defined as follows:
$$
\Gamma_{0}\psi(\mathcal{B}){\bf f}(x)=\psi({\mathcal B}){\bf f}(0)e^{-x},  \qquad \Gamma_{1}\psi(\mathcal{B}){\bf f}(x)=2\psi(\mathcal{B})[{\bf f}'(0)+{\bf f}(0)]e^{-x}.
$$
\end{corollary}
\begin{proof}
It is well known  (see, e.g., \cite{CK})  that the positive boundary triplet $({\mathcal H}', \Gamma_{0}',\Gamma_{1}')$
of ${{\mathcal B}^*}^2$ has the form:
${\mathcal H}'=\left\{\left[\begin{array}{c}
{\alpha}  \\
{\beta}
\end{array}\right]e^{-x}  :   \alpha, \beta\in\mathbb{C} \right\}$  and
 $$
\Gamma_0'{\bf f}={\bf f}(0)e^{-x}, \qquad  \Gamma_1{\bf f}=2[{\bf f}'(0)+{\bf f}(0)]e^{-x},  \qquad {\bf f}\in{W_2^2(\mathbb{R}_+, \mathbb{C}^2)}.
$$
Applying Lemma \ref{AGH678} we complete the proof.
\end{proof}

\subsection{The $S$-matrix for positive self-adjoint ${\bf H}_{a{\bf q}}$}
 \begin{theorem}\label{ups222}
The $S$-matrix for positive self-adjoint operator ${\bf H}_{a{\bf q}}$ has the form
\begin{equation}\label{ups4}
S(z)=\left[\begin{array}{cc}
\Psi_1(z)  &  0 \\
0 &  \Psi_2(z) \end{array}\right]-\frac{2zi}{a-W(z^2)}\left[\begin{array}{cc}
c(z, q_1)\overline{c(-\overline{z}, q_1)}    &  c(z, q_1)\overline{c(-\overline{z}, q_2)} \vspace{3mm} \\
c(z, q_2)\overline{c(-\overline{z}, q_1)}   &  c(z, q_2)\overline{c(-\overline{z}, q_2)} 
\end{array}
\right],
\end{equation}
where $c(\mu, q_i)$ are determined by \eqref{AGH44b} 
and $\Psi_j(z)$ are  holomorphic  continuations of the functions ${\psi_j(-\delta)}/{\psi_j(\delta)}$ $(\delta\in\mathbb{R})$ into $\mathbb{C}_-$  
such that $|\Psi_j(z)|<1$ and $\overline{\Psi_j(z)}=\Psi_j(-\overline{z})$.
\end{theorem}
\begin{proof}
By Theorem \ref{esse3}, for the calculation of $S$-matrix, one need to
 find operators $C(z)$ in \eqref{tttt}. To do that we analyze vectors 
 $$
 {\bf u}{\in}P_{{\mathfrak H}_0}({{\bf H}_{a{\bf q}}}-z^2{I})^{-1}\ker(B^*+\overline{z}I)
 $$
 in more detail.    First of all we note that $\ker(B^*+\overline{z}I)=\psi({\mathcal B})\{{\bf h}_{-\overline{z}}\}$ by Lemma \ref{GGG24}.
 Consider the equation\footnote{The coefficient $(\overline{z}^2-z^2)$ is used for the simplification of formulas below.}
\begin{equation}\label{AGH101}
({\bf H}_{a{\bf q}}-z^2{I}){\bf f}=(\overline{z}^2-z^2)\psi({\mathcal B}){\bf h}_{-\overline{z}},  \qquad  z\in\mathbb{C}_-\setminus{i\mathbb{R}_-}.
\end{equation}
Its solution ${\bf f}\in\mathcal{D}({\bf H}_{a{\bf q}})$ is determined uniquely and 
\begin{equation}\label{AGHNEW1}
{\bf u}=P_{{\mathfrak H}_{0}}{\bf f}=(\overline{z}^2-z^2)P_{{\mathfrak H}_{0}}({\bf H}_{a{\bf q}}-z^2{I})^{-1}\psi({\mathcal B}){\bf h}_{-\overline{z}}
\end{equation}
belongs to $\mathcal{D}({B^*}^2)$ due to \eqref{GGG55}.
In view of \eqref{AGH23b},  ${\bf u}=\psi({\mathcal B}){\bf v}$, where ${\bf v}\in{W_2^2(\mathbb{R}_+, \mathbb{C}^2)}$ and
${B^*}^2\psi({\mathcal B}){\bf v}=\psi({\mathcal B}){\mathcal{B}^*}^2{\bf v}$.
Moreover, since $P_{\mathfrak{H}_0}=\psi(\mathcal{B})\psi(\mathcal{B})^*$, the relation \eqref{AGHNEW1} yields
\begin{equation}\label{AGHNEW1BB}
{\bf v}=(\overline{z}^2-z^2){\psi}(\mathcal{B})^*({\bf H}_{a{\sf q}}-z^2{I})^{-1}{\psi}(\mathcal{B}){\bf h}_{-\overline{z}}.
\end{equation}

Applying $P_{{\mathfrak H}_{0}}$ to the both parts of \eqref{AGH101} and using \eqref{GGG55} we obtain
$$
({B^*}^2-z^2I){\bf u}=\psi({\mathcal B})({{\mathcal B}^*}^2-z^2I){\bf v}=(\overline{z}^2-z^2)\psi({\mathcal B}){\bf h}_{-\overline{z}}.
$$
Therefore,  $({{\mathcal B}^*}^2-z^2I){\bf v}=(-\frac{d^2}{dx^2}-z^2I){\bf v}=(\overline{z}^2-z^2){\bf h}_{-\overline{z}}$. 
This means that 
\begin{equation}\label{AGHNEW1CC}
{\bf v}={\bf h}_{-\overline{z}}+{\bf h}_{z},  \qquad {\bf u}=\psi({\mathcal B}){\bf v}=\psi({\mathcal B}){\bf h}_{-\overline{z}}+\psi({\mathcal B}){\bf h}_{z},
\end{equation}
where ${\bf h}_{z}\in\ker(B^*-zI)$ is determined uniquely by the choice of ${\bf h}_{-\overline{z}}$.
Applying operators $\Gamma_i$ from Corollary \ref{kkkk} we obtain 
$$
\Gamma_0{\bf u}=\psi({\mathcal B})\left[\begin{array}{c}
\alpha_{-\overline{z}}+\alpha_z  \\
\beta_{-\overline{z}}+\beta_z
\end{array}\right]e^{-x}, \quad \Gamma_1{\bf u}=2\psi({\mathcal B})\left[\begin{array}{c}
(1+i\overline{z})\alpha_{-\overline{z}}+(1-iz)\alpha_z  \\
(1+i\overline{z})\beta_{-\overline{z}}+(1-iz)\beta_z
\end{array}\right]e^{-x}. 
$$
Since $\dim{\mathcal H}=2$, the function  ${C}(z)$ in Theorem \ref{esse3} is  
$2\times{2}$-matrix-valued.  The substitution of $\Gamma_i{\bf u}$ into the characteristic relation \eqref{tttt} gives
$$
2{C}(z)\left[\begin{array}{c}
(1+i\overline{z})\alpha_{-\overline{z}}+(1-iz)\alpha_z  \\
(1+i\overline{z})\beta_{-\overline{z}}+(1-iz)\beta_z
\end{array}\right]=\left[\begin{array}{c}
\alpha_{-\overline{z}}+\alpha_z  \\
\beta_{-\overline{z}}+\beta_z
\end{array}\right]
$$
and, after elementary transformations, 
\begin{equation}\label{AGH35c}
[I-2(1-iz){C}(z)]^{-1}\left[\begin{array}{c}
\alpha_{-\overline{z}} \\
\beta_{-\overline{z}}
\end{array}\right]=\frac{1}{2i{Re \ z}}\left[\begin{array}{c}
(1+i\overline{z})\alpha_{-\overline{z}} + (1-iz)\alpha_z \\
(1+i\overline{z})\beta_{-\overline{z}} + (1-iz)\beta_z
\end{array}\right]. 
\end{equation}
The substitution of \eqref{AGH35c} into \eqref{red1} gives the $S$-matrix
\begin{equation}\label{ggg11d}
S(z)\left[\begin{array}{c}
\alpha_{-\overline{z}} \\
\beta_{-\overline{z}}
\end{array}\right]=-i\frac{Im \ z}{Re \ z}\left[\begin{array}{c}
\alpha_{-\overline{z}} \\
\beta_{-\overline{z}} 
\end{array}\right]-\frac{z}{Re \ z}\left[\begin{array}{c}
\alpha_z\\
\beta_z
\end{array}\right], \qquad z\in\mathbb{C}_-\setminus{i\mathbb{R}_-}.
\end{equation}
Here  $\alpha_z, \beta_z$ are functions of  parameters $\alpha_{-\overline{z}}, \beta_{-\overline{z}}\in\mathbb{C}$.
Indeed, in view of  \eqref{AGHNEW1BB} and  \eqref{AGHNEW1CC}
${\bf h}_{z}=-{\bf h}_{-\overline{z}}+(\overline{z}^2-z^2){\psi}(\mathcal{B})^*({\bf H}_{a{\bf q}}-z^2{I})^{-1}{\psi}(\mathcal{B}){\bf h}_{-\overline{z}}$ and hence,  
\begin{equation}\label{ups}
\left[\begin{array}{c}
\alpha_{z} \\
\beta_{z} 
\end{array}\right]e^{-izx}=(-I+ (\overline{z}^2-z^2){\psi}(\mathcal{B})^*({\bf H}_{a{\bf q}}-z^2{I})^{-1}{\psi}(\mathcal{B}))\left[\begin{array}{c}
\alpha_{-\overline{z}} \\
\beta_{-\overline{z}} 
\end{array}\right]e^{i\overline{z}x},  
\end{equation}

The $S$-matrix $S(z)$ depends on the choice of ${\bf H}_{a{\sf q}}$. If ${\bf H}_{a{\bf q}}={\bf H}_{\infty}$, then
this  operator is a positive self-adjoint extension of the symmetric operators $\mathcal{B}^2$ and $B^2$.
 By Theorem \ref{new31b} one can construct  two pairs of subspaces $D_\pm$ that are determined by $\mathcal{B}$ and $B$, respectively. 
Therefore,  one can define two $S$-matrices $S_1(\cdot)$ and $S(\cdot)$ for ${\bf H}_{\infty}$
corresponding  to the cases where ${\bf H}_{\infty}$ is considered as an extension of $\mathcal{B}^2$ or an extension of $B^2$.
The both of $S$-matrices are defined by  \eqref{red1} but,  in the first case,  $C(z)=0$ and, therefore $S_1(z)=\sigma_0$.  
In view of \cite[Proposition 3.1]{Gawlik},
\begin{equation}\label{ups23}
S(z)=\left[\begin{array}{cc}
\Psi_1(z)  &  0 \\
0 &  \Psi_2(z) \end{array}\right]S_1(z)=\left[\begin{array}{cc}
\Psi_1(z)  &  0 \\
0 &  \Psi_2(z) \end{array}\right],
\end{equation}
where $\Psi_j(z)$ are  holomorphic  functions in $\mathbb{C}_-$ such that $|\Psi_j(z)|<1$ and $\overline{\Psi_j(z)}=\Psi_j(-\overline{z})$.
Moreover, the boundary values of $\Psi_j(z)$ on $\mathbb{R}$ coincide with ${\psi_j(-\delta)}/{\psi_j(\delta)}$. 

Due to \eqref{ups}, the coefficients $\alpha_z, \beta_z$ in \eqref{ggg11d} depend
on the choice of ${\bf H}_{a{\bf q}}$. 
The resolvent formula \eqref{DDD4} and \eqref{ups} allow one
to present  $\alpha_z=\alpha_z({\bf H}_{a{\bf q}}), \ \beta_z=\beta_z({\bf H}_{a{\bf q}})$ as the sum of 
$\alpha_z({{\bf H}_\infty}),  \beta_z({{\bf H}_\infty})$ 
and a function that is determined by the difference between  $({\bf H}_{a{\bf q}}-z^2{I})^{-1}$ and $({\bf H}_{\infty}-z^2{I})^{-1}$  
 (see the second part in \eqref{DDD4}). 
Such decomposition and  \eqref{ups23} allows one to rewrite \eqref{ggg11d}: 
\begin{equation}\label{ggg11c}
S(z)\left[\begin{array}{c}
\alpha_{-\overline{z}} \\
\beta_{-\overline{z}}
\end{array}\right]=\left[\begin{array}{c}
\Psi_1(z)\alpha_{-\overline{z}} \\
\Psi_2(z)\beta_{-\overline{z}} 
\end{array}\right]-\frac{ze^{izx}}{Re \ z}(\overline{z}^2-z^2)\frac{({\bf h}_{-\overline{z}}, \psi({\mathcal B})^*{\bf u}_{-\overline{z}})_{+}}{a-W(z^2)}{\psi}(\mathcal{B})^*{\bf u}_{z}.
\end{equation}

In view of \eqref{AGH79} with $\mu=-\overline{z}$ 
$$
\frac{(\overline{z}^2-z^2)({\bf h}_{-\overline{z}}, \psi({\mathcal B})^*{\bf u}_{-\overline{z}})_{+}}{Re \ z}=2i\left\langle\left[\begin{array}{c}
\alpha_{-\overline{z}} \\
\beta_{-\overline{z}}
\end{array}\right], \left[\begin{array}{c}
c(-\overline{z}, q_1) \\
c(-\overline{z}, q_2)
\end{array}\right]\right\rangle,
$$
where $\langle\cdot, \cdot\rangle$ is the inner product in $\mathbb{C}^2$.
Substituting this expression into \eqref{ggg11c}  and using \eqref{AGH79}  with $\mu=z$, we obtain
$$
S(z)\left[\begin{array}{c}
\alpha_{-\overline{z}} \\
\beta_{-\overline{z}}
\end{array}\right]=\left[\begin{array}{c}
\Psi_1(z)\alpha_{-\overline{z}} \\
\Psi_2(z)\beta_{-\overline{z}} 
\end{array}\right]-\frac{2zi}{a-W(z^2)}\left\langle\left[\begin{array}{c}
\alpha_{-\overline{z}} \\
\beta_{-\overline{z}}
\end{array}\right], \left[\begin{array}{c}
c(-\overline{z}, q_1) \\
c(-\overline{z}, q_2)
\end{array}\right]\right\rangle\left[\begin{array}{c}
c(z, q_1) \\
c(z, q_2)
\end{array}\right].
$$
A rudimentary linear algebra exercise leads to the conclusion this formula for $S(z)$
can be rewritten as   \eqref{ups4} for  $z\in\mathbb{C}_-\setminus{i\mathbb{R}_-}$.
Since the $S$-matrix is holomorphic in the lower half-plain,  the formula
\eqref{ups4} remains true for $\mathbb{C}_-$.
\end{proof}
 
The expression \eqref{ups4} is based on the Krein-Naimark resolvent formula \eqref{DDD4} and it
 allows one to establish various useful relationships between $S$-matrix
and the operator ${\bf H}_{a{\bf q}}$.  An alternative formula for $S$-matrix in terms of reflection and transmission coefficients
is presented below.

By virtue of Lemma \ref{AGH44},
\begin{equation}\label{JJJ1}
P_{\mathfrak{H}_0}\left[\begin{array}{c}
 e^{i\overline{z}x} \\
0
\end{array}\right]=\psi(\mathcal{B})\psi(\mathcal{B})^*\left[\begin{array}{c}
 e^{i\overline{z}x} \\
0
\end{array}\right]=\psi(\mathcal{B})\left[\begin{array}{c}
 \overline{\psi_1(-z)} \\
0
\end{array}\right]e^{i\overline{z}x}
\end{equation}
and, similarly,  $P_{\mathfrak{H}_0}\left[\begin{array}{c}
{\alpha}_z  \\
{\beta}_z
\end{array}\right]e^{-izx}=\psi(\mathcal{B})\left[\begin{array}{c}
{\alpha}_z\overline{\psi_1(\overline{z})}  \\
{\beta}_z\overline{\psi_2(\overline{z})}
\end{array}\right]e^{-izx}$.

Setting ${\bf h}_{-\overline{z}}=\left[\begin{array}{c}
\overline{\psi_1(-z)}  \\
0
\end{array}\right]e^{i\overline{z}x}$ in \eqref{AGH101} and using \eqref{JJJ1} we obtain
$$
({\bf H}_{a{\bf q}}-z^2{I}){\bf f}=(\overline{z}^2-z^2)\psi(\mathcal{B}){\bf h}_{-\overline{z}}=(\overline{z}^2-z^2)P_{\mathfrak{H}_0}\left[\begin{array}{c}
 e^{i\overline{z}x} \\
0
\end{array}\right],  \ z\in\mathbb{C}_-\setminus{i\mathbb{R}_-}
$$
and, in view of \eqref{AGHNEW1},  \eqref{AGHNEW1CC},  its solution ${\bf f}$ satisfies the relation 
$$
P_{\mathfrak{H}_0}{\bf f}=\psi(\mathcal{B})\left[\begin{array}{c}
\overline{\psi_1(-z)}  \\
0
\end{array}\right]e^{i\overline{z}x}+\psi(\mathcal{B})\left[\begin{array}{c}
{\alpha}_{z}  \\
{\beta}_z
\end{array}\right]e^{-izx}=P_{\mathfrak{H}_0}\left[\begin{array}{c}
e^{i\overline{z}x}+R_z^1e^{-izx}  \\
T_z^1e^{-izx}
\end{array}\right], 
$$
where 
$$
R_z^1=\frac{\alpha_z}{\overline{\psi_1(\overline{z})}},  \qquad  T_z^1=\frac{\beta_z}{\overline{\psi_2(\overline{z})}}
$$
are called \emph{the reflection} and \emph{the transmission} coefficients, respectively.

Similarly, assuming ${\bf h}_{-\overline{z}}=\left[\begin{array}{c}
0 \\
\overline{\psi_2(-z)}
\end{array}\right]e^{i\overline{z}x}$ and considering the solution ${\bf f}$ of  
$$
({\bf H}_{a{\bf q}}-z^2{I}){\bf f}=(\overline{z}^2-z^2)P_{\mathfrak{H}_0}\left[\begin{array}{c}
0 \\
 e^{i\overline{z}x}
\end{array}\right],
$$
 we obtain
$$
P_{\mathfrak{H}_0}{\bf f}=P_{\mathfrak{H}_0}\left[\begin{array}{c}
T_z^2e^{-izx}  \\
e^{i\overline{z}x}+R_z^2e^{-izx}
\end{array}\right],  \qquad  R_z^2=\frac{{\beta}_z}{\overline{\psi_2(\overline{z})}}, \quad T_z^2=\frac{{\alpha}_z}{\overline{\psi_1(\overline{z})}}. 
$$

The reflection $R_z^j$ and the transmission  $T_z^j$ coefficients described above allow one to obtain an alternative 
 formula for $S$-matrix. 

\begin{theorem}\label{ups2}
The $S$-matrix of a positive self-adjoint operator
${\bf H}_{a{\bf q}}$ has the form 
 \begin{equation}\label{AGHHH}
S(z)=\frac{-z}{Re\ z}\left[\begin{array}{cc}
\theta_{11}(z)R_z^1+i\frac{Im \ z}{z} &  \theta_{12}(z)T_z^2 \vspace{3mm} \\
\theta_{21}(z)T_z^1  & \theta_{22}(z)R_z^2+i\frac{Im \ z}{z}
\end{array}\right],  \quad  \theta_{nm}(z)=\frac{\overline{\psi_n(\overline{z})}}{\overline{\psi_m(-z)}}.
\end{equation}
\end{theorem}
\begin{proof}
Setting in \eqref{ggg11d}: 
$$
\alpha_{-\overline{z}}=\overline{\psi_1(-z)}, \quad \beta_{-\overline{z}}=0, \quad \alpha_z=\overline{\psi_1(\overline{z})}R_z^1, \quad
 \beta_z=\overline{\psi_2(\overline{z})}T_z^1 
 $$
and
$$
\alpha_{-\overline{z}}=0, \quad  \beta_{-\overline{z}}=\overline{\psi_2(-z)}, \quad \alpha_z=\overline{\psi_1(\overline{z})}T_z^2, \quad 
\beta_z=\overline{\psi_2(\overline{z})}R_z^2
$$
we obtain a system of four linear equations with respect to unknowns coefficients  of the  $S$-matrix  $S(z)=\left[\begin{array}{cc}
s_{11}  &  s_{12} \\
s_{21}  & s_{22}
\end{array}\right]$. Its solution gives rise to \eqref{AGHHH} for all $z\in\mathbb{C}_-\setminus{i\mathbb{R}_-}$.
Since $S(z)$ is holomorphic in  $\mathbb{C}_-$, the formula
\eqref{AGHHH} holds for all $z\in\mathbb{C}_-$.
\end{proof}

\subsubsection{Example of ordinary $\delta$-interaction}
In view of  \eqref{newww6a},  the ordinary $\delta$-interaction
corresponds to ${\bf q}=0$. The operators ${\bf H}_{a}={\bf H}_{a0}=-\frac{d^2}{dx^2}$ have the domains:
$$
\mathcal{D}({\bf H}_{a{\bf q}})=\{{\bf f}\in W_2^2(\mathbb{R}_+,\mathbb{C}^2): \
 [{\bf f}]_s=0,  \quad  [{\bf f'}]_r= {a}[{\bf f}]_r \}.
$$
The function ${\bf q}=0$ is non-cyclic and one can set $ \psi_1=\psi_2={1}$.
Then $P_{\mathfrak{H}_0}=I$ and the reflection and the transmission coefficients are determined as follows:
$$
R_z^1=R_z^2=\frac{-a+i(\overline{z}-z)}{a+2iz}, \qquad T_z^1=T_z^2=\frac{2i Re \ z}{a+2iz}.
$$ 
Substituting the obtained expressions in  \eqref{AGHHH} and taking into account that $\theta_{nm}(z)=1$, we obtain
a matrix-valued $S$-function
\begin{equation}\label{GH}
S(z)=\frac{1}{a+2iz}\left[\begin{array}{cc}
a &  -2iz \vspace{3mm} \\
-2iz  & a
\end{array}\right],
\end{equation}
which is holomorphic  on $\mathbb{C}_-$ for positive self-adjoint operators  ${\bf H}_{a}$ (the positivity of ${\bf H}_{a}$ is distinguished by
the condition $a\geq{0}$).

The same formula \eqref{GH} can be deduced from \eqref{ups4} if one take into account that
$\Psi_j={1}$ since $\psi_j={1}$ and $W(z^2)=-2iz$,  $c(z, q_j)=1$ by virtue of 
 \eqref{DDD3} and \eqref{AGH44b}, respectively.

\section{Operators ${\bf H}_{a{\bf q}}$ and their  $S$-matrices}\label{sec.5}
The example above leads to a natural assumption that
the formulas \eqref{ups4}, \eqref{AGHHH} allow to construct 
a function $S(z)$ for each operator ${\bf H}_{a{\bf q}}$
(assuming, of course, that ${\bf q}$ is non-cyclic). 
We will call it  \emph{the $S$-matrix} of ${\bf H}_{a{\bf q}}$.  If ${\bf H}_{a{\bf q}}$ is positive self-adjoint, then
the $S$-matrix is the consequence of proper arguments of the Lax-Phillips theory and it coincides with the analytical continuation of the Lax-Phillips scattering matrix
into $\mathbb{C}_-$.  Otherwise,  $S(z)$  is  defined directly by \eqref{ups4}, \eqref{AGHHH}  and it can be considered as a characteristic function of  ${\bf H}_{a{\bf q}}$. 
In this section, we describe properties of  ${\bf H}_{a{\bf q}}$ in terms of the corresponding $S$-matrix.

It follows from \eqref{ups4} that a $S$-matrix of ${\bf H}_{a{\bf q}}$ is a meromorphic matrix-valued function on $\mathbb{C}_-$.
Its poles describe the point spectrum of ${\bf H}_{a{\bf q}}$ in $\mathbb{C}\setminus{[0, \infty)}$. 
\begin{lemma}\label{nnn1}
If  $z\in\mathbb{C_-}$ is a pole of $S(z)$, then $z^2$ belongs to the point spectrum of ${\bf H}_{a{\bf q}}$.
\end{lemma}
\begin{proof}
By virtue of  \eqref{ups4}, if $z\in\mathbb{C_-}$ is a pole for $S(z)$ then $a=W(z^2)$. This identity means that  
$z^2\in\sigma_p({\bf H}_{a{\bf q}})$
because ${\bf H}_{a{\bf q}}$ is defined by \eqref{DDD1} and $W(z^2)$ is
the Weyl-Titchmarsh function associated to  the boundary triplet  $(\mathbb{C}, \Gamma_0,
\Gamma_1)$ (see Sec. \ref{sec3.1} and \cite[Proposition 14.17]{Schm}).
\end{proof}

\begin{remark}
It may happen that the $S$-matrix `does not hear' an eigenvalue $z^2$. This is the case where the corresponding 
eigenfunction ${\bf u}_z$ is orthogonal to $\psi(\mathcal{B})L_2(\mathbb{R}_+, \mathbb{C}^2)$ and, as a result,
the coefficients  $c(z, q_i)$  vanish, see Sec. \ref{sec.5.1.1}.
\end{remark}

Divide the half-plane $\mathbb{C}_-$ into three parts
$$
\mathbb{C}_-^-=\{z  :  Re \ z<0 \}; \quad  \mathbb{C}_-^0=\{z  :  Re \ z=0 \}; \quad \mathbb{C}_-^+=\{z  :  Re \ z>0 \}.
$$
\begin{lemma}\label{nnn2}
If $S(z)$ has a pole in $\mathbb{C}_-^\mp$, then $S(z)$ has to be analytical on the opposite part $\mathbb{C}_-^\pm$.
If $S(z)$ has a pole on the middle part $\mathbb{C}_-^0$, then $S(z)$ is analytical on $\mathbb{C}_-^-\cup\mathbb{C}_-^+$
 and  ${\bf H}_{a{\bf q}}$ is  a self-adjoint operator.
\end{lemma}
\begin{proof}
Let  $z\in\mathbb{C_-^-}$ be a pole for $S(z)$.  By virtue of  \eqref{ups4},  $a=W(z^2)$, where $Im\ z^2>0$ and $Im \ a>0$ since $Im \ W(z^2)/Im \ z^2>0$
\cite[Sec. 14.5]{Schm}.  Similar arguments for a pole $z\in\mathbb{C_-^+}$ lead to the conclusion that $Im \ a<0$. 
The obtained contradiction means that the existence of a pole in $\mathbb{C_-^+}$ ($\mathbb{C_-^-}$) implies  the absence of poles in $\mathbb{C_-^-}$ ($\mathbb{C_-^+}$).  

If $z\in\mathbb{C}_-^0$ is a pole, then ${\bf H}_{a{\bf q}}$ has a negative eigenvalue and ${\bf H}_{a{\bf q}}$ has to be self-adjoint  due to \cite[Corollary 5.2]{KZ}.
\end{proof}

An eigenvalue $z^2\in\mathbb{C}\setminus{[0, \infty)}$ of ${\bf H}_{a{\bf q}}$ is called \emph{an exceptional point} if its
geometrical multiplicity does not coincide with the algebraic one.  The presence of an exceptional point means that ${\bf H}_{a{\bf q}}$
cannot be self-adjoint for any choice of inner product. It follows from Lemma \ref{nnn2} that an exceptional point $z^2$ is necessarily non-real and
$z\in\mathbb{C}_-^-\cup\mathbb{C}_-^+$.

\begin{lemma}\label{nnn3}
A non-simple pole\footnote{a pole of order greater then  one} $z$ of $S(z)$  
corresponds to an exceptional point $z^2$ of ${\bf H}_{a{\bf q}}$. 
\end{lemma}
\begin{proof}
A non-simple pole $z$ of $S(z)$  means that the function $(a-W(\lambda))^{-1}$ has a non-simple pole for $\lambda=z^2$.
This yields that $W'(z^2)=0$, where $W'(\lambda)=dW/d\lambda$.
In view of \cite[Theorem 5.4]{KZ}, an eigenvalue $z^2$ of ${\bf H}_{a{\bf q}}$ is an exceptional point if and only if 
$W'(z^2)=0$. 
\end{proof}

\begin{lemma}\label{AGH789}
Let $S_{{\bf H}_{a{\bf q}}}(z)$ be a $S$-matrix of ${\bf H}_{a{\bf q}}$. Then
$$
S^*_{{\bf H}_{a{\bf q}}}(z)=S_{{\bf H}_{\overline{a}{\bf q}}}(-\overline{z})=S_{{\bf H}^*_{{a}{\bf q}}}(-\overline{z}).
$$
\end{lemma}
\begin{proof} Using \eqref{ups4} for the calculation of the adjoint, we get   
$$
S^*_{{\bf H}_{a{\sf q}}}(z)=\left[\begin{array}{cc}
\overline{\Psi_1({z})}  &  0 \\
0 &  \overline{\Psi_2(z)} \end{array}\right]+\frac{2\overline{z}i}{\overline{a}-\overline{W(z^2)}}\left[\begin{array}{cc}
c(-\overline{z}, q_1)\overline{c(z, q_1)}    &  c(-\overline{z}, q_1)\overline{c(z, q_2)} \vspace{3mm} \\
c(-\overline{z}, q_2)\overline{c(z, q_1)}   &  c(-\overline{z}, q_2)\overline{c(z, q_2)} 
\end{array}
\right].
$$     
In  view of Theorem \ref{ups222} $\overline{\Psi_j(z)}=\Psi_j(-\overline{z})$. 
Moreover, $\overline{W(z^2)}=W((-\overline{z})^2)$. This well-known property of the Weyl-Titchmarsh functions \cite[Chap. 14]{Schm} can easily
be derived from \eqref{DDD3}. Taking these facts into account and using  \eqref{ups4} for the calculation of $S_{{\bf H}_{\overline{a}{\bf q}}}(-\overline{z})$, 
we arrive at the conclusion that $S^*_{{\bf H}_{a{\bf q}}}(z)=S_{{\bf H}_{\overline{a}{\bf q}}}(-\overline{z})$.
Now, to complete the proof it suffices to remark that ${\bf H}^*_{a{\bf q}}={\bf H}_{\overline{a}{\sf q}}$  due to \eqref{DDD1} and \cite[Lemma 14.6]{Schm}. 
\end{proof}
\begin{corollary}\label{nnn125}
Let $S(z)$ be a $S$-matrix of ${\bf H}_{a{\bf q}}$. 
Then  ${\bf H}_{a{\bf q}}$ is self-adjoint if and only if  ${S^*(z)}=S(-\overline{z})$.
\end{corollary}
\begin{proof} If ${\bf H}_{a{\bf q}}$ is self-adjoint, then $a\in\mathbb{R}$ and
${S^*(z)}=S(-\overline{z})$ due to Lemma \ref{AGH789}.  
Conversely, as follows from the proof above, the relation ${S^*(z)}=S(-\overline{z})$ is possible only
in the case of real $a$. This implies the self-adjointness of ${\bf H}_{a{\bf q}}$.
\end{proof}
  
\subsection{Examples}\label{sec.5.1}
\subsubsection{Even function $q$ with finite support.}\label{sec.5.1.1}
We consider the simplest  example of even function with finite support
$$
q(x)=M\chi_{[-\rho,\rho]}(x), \qquad M\in\mathbb{C}, \quad \rho>0.
$$
In this case, $Yq={\bf q}=M\left[\begin{array}{c}
\chi_{[0,\rho]}(x) \\
\chi_{[0,\rho]}(x) 
\end{array}\right]$.

Denote $\psi(\delta)=e^{i\delta\rho}$.  The function $\psi$ belongs to $H^\infty(\mathbb{C}_+)$ and
the operator $\psi(\mathcal B)$ in \eqref{AGHH} acts in $L_2(\mathbb{R}_+)$ as follows:
\begin{equation}\label{HHH1}
\psi(\mathcal B)f=\left\{ \begin{array}{ll}
f(x-\rho) & \textrm{for $x\geq\rho$}\\
0 & \textrm{for $ x < \rho $}
\end{array} \right.
\end{equation}
Further, we extend the action of $\psi(\mathcal B)$ onto $L_2(\mathbb{R}_+, \mathbb{C}^2)$ 
assuming in \eqref{AGH77} that $\psi_1(\mathcal{B})=\psi_2(\mathcal{B})=\psi(\mathcal{B})$.
It follows from \eqref{HHH1} that $\psi(\mathcal B)^*{\bf f}={\bf f}(x+\rho)$. Hence, 
\begin{equation}\label{HHH2}
P_{\mathfrak{H}_0}{\bf f}=\psi(\mathcal B)\psi(\mathcal B)^*{\bf f}=\left\{ \begin{array}{ll}
{\bf f}(x) & \textrm{for $x\geq\rho$}\\
0  & \textrm{for $ x < \rho $}
\end{array} \right.
\end{equation}

The formula \eqref{HHH2} and Lemma \ref{zzzz} imply that ${\bf q}$ is non-cyclic.  
Therefore, for ${\bf H}_{a{\bf q}}$ there exists a $S$-matrix defined by \eqref{ups4}. 
Let us specify the counterparts of \eqref{ups4}.  First of all we note that $\Psi_1(z)=\Psi_2(z)=e^{-2iz\rho}$ 
as the holomorphic continuation of $e^{-2i\delta\rho}=\frac{\psi(-\delta)}{\psi(\delta)}$ into $\mathbb C_-$.
Further, in view of \eqref{newww2},
$$
({\bf H}_\infty-\mu^2I)^{-1}{\bf q}=-\frac{M}{2\mu^2}[(e^{-i\mu\rho}+e^{i\mu{m}(x)}-2){\bf e}^{-i\mu{x}}+(e^{-i\mu{m}(x)}-e^{-i\mu\rho}){\bf e}^{i\mu{x}}],    
$$
where $m(x)=\min\{x, \rho\}$ and $\mu\in\mathbb{C}_-$. This formula and \eqref{AGH44b} lead to the conclusion that
$$
c(\mu, q_1)=c(\mu, q_2)=e^{-i\mu\rho}\left(1-\kappa_{\mu}\frac{M}{\mu^2}\right), \qquad \kappa_{\mu}=1-\cos\mu\rho.
$$
Our next step is the calculation of $W(z^2)$ using formula $\eqref{DDD3}$ and the expression for $({\bf H}_\infty-\mu^2I)^{-1}$, that gives 
$$
W(z^2)=-2iz-\frac{4Re \ M}{iz}(1-e^{-iz\rho})+\frac{|M|^2}{iz^3}\big[(e^{-iz\rho}-2)^2-2iz\rho-1 \big].
$$
Substituting the expressions obtained above  into \eqref{ups4} we find the $S$-matrix for ${\bf H}_{a{\bf q}}$ 
$$
S(z)=e^{-2iz\rho}\left(\sigma_0 -\frac{2i(z^2-\kappa_{z}{M})(z^2-\kappa_{z}{\overline{M}})}{z^3(a-W(z^2))}
\left[ \begin{array}{cc}
1 &  1 \\
1 &  1
\end{array} \right]\right).
$$

Let us assume that $z_0\in\mathbb{C}_-$ satisfies the relation $z^2_0-\kappa_{z_0}{M}=0$ and $W'(z_0^2)\not=0$.
Set $a=W(z_0^2)$. Then the operator
${\bf H}_{a{\bf q}}$ has the eigenvalue $z_0^2$ with eigenfunction  ${\bf u}_{z_0}$.  It follows from \eqref{DDD} and the explicit expression
for $({\bf H}_\infty-\mu^2I)^{-1}$ that
$$
{\bf u}_{z_0}=\frac{1-\cos{z_0}(\rho-x)}{z_0^2}{\bf q}.
$$
In view of \eqref{HHH2}, the eigenfunction ${\bf u}_{z_0}$ is orthogonal to $\mathfrak{H_0}$ and it has no impact on the 
$S$-matrix  $S(z)$  (no pole for $z=z_0$).  

\subsubsection{Odd function $q$ with finite support.}
Similarly to the previous case, we consider the  odd function 
$$
q(x)=M{\rm sign}(x)\chi_{[-\rho,\rho]}(x), \qquad M\in\mathbb{C}, \quad \rho>0.
$$
 In this case,  ${\bf q}=M\left[\begin{array}{c}
\chi_{[0,\rho]}(x) \\
-\chi_{[0,\rho]}(x) 
\end{array}\right]$ is non-cyclic and it is orthogonal to the same subspace $\mathfrak{H}_0=\psi(\mathcal{B})L_2(\mathbb{R}_+, \mathbb{C}^2)$ 
as above. Further,
$$
c(\mu, q_1)=e^{-i\mu\rho}\left(1-\kappa_{\mu}\frac{M}{\mu^2}\right),\qquad c(\mu, q_2)=e^{-i\mu\rho}\left(1+\kappa_{\mu}\frac{M}{\mu^2}\right)
$$
and  $W(z^2)=-2iz+\frac{|M|^2}{iz^3}\big[(e^{-iz\rho}-2)^2-2iz\rho-1 \big].$ Then \eqref{ups4} takes the form:  
$$
S(z)=e^{-2iz\rho}\left(\sigma_0 -\frac{2zi}{a-W(z^2)}
\left[ \begin{array}{cc}
1-\kappa_z\frac{2\textrm{Re}M}{z^2}+\kappa_z^2\frac{|M|^2}{z^4} &  1-\kappa_z\frac{2\textrm{Im}M}{z^2}-\kappa_z^2\frac{|M|^2}{z^4}  \\
1+\kappa_z\frac{2\textrm{Im}M}{z^2}-\kappa_z^2\frac{|M|^2}{z^4}  &  1+\kappa_z\frac{2\textrm{Re}M}{z^2}+\kappa_z^2\frac{|M|^2}{z^4} 
\end{array} \right]\right).
$$
It is easy to see that the entries of the last matrix can not vanish simultaneously. This means 
that $z\in\mathbb{C}_-$ is a pole of $S(z)$ if and only if $a=W(z^2)$. Therefore, in contrast to Sec. \ref{sec.5.1.1},
the poles of $S(z)$ completely determine the point spectrum of ${\bf H}_{a{\bf q}}$ in $\mathbb{C}\setminus\mathbb{R}_+$. 

 \subsubsection{Functions $q$ with infinite support.}
The range of applicability of our results is not limited to operators ${\bf H}_{a{\bf q}}$, where 
${\bf q}=Yq$ has finite support.  Due to Lemma \ref{zzzz} and Theorem \ref{yyyy}, the $S$-matrix
\eqref{ups4} can be constructed for an operator ${\bf H}_{a{\bf q}}$  when  ${\bf q}$ is non-cyclic with respect to
the backward shift operator $T^*$ in $L_2(\mathbb{R}_+, \mathbb{C}^2)$.
Various examples of non-cyclic functions
can be found in \cite{DSS, MFAT1}. Consider, for instance, the function
$q(x)=P_{m}(x)e^{-|x|}$, where  $P_{m}$ is a polynomial of order $m$. 
Then  
$$
{\bf q}=\left[\begin{array}{c}
P_{m}(x) \\
P_{m}(-x) 
\end{array}\right]e^{-x}, \qquad x\geq{0}.
$$
 Decompose the functions $P_{m}(\pm{x})e^{-x}\in{L_2(\mathbb{R}_+)}$: 
\begin{equation}\label{FFFF}
e^{-x}P_{m}(x)=\sum_{n=0}^{m}c_n{q}_{n}(2x), \qquad  e^{-x}P_{m}(-x)=\sum_{n=0}^{m}d_n{q}_{n}(2x),
\end{equation}
 with respect to the orthonormal basis  of the Laguerre functions 
 $$
 {q}_n(x)=\frac{e^{x/2}}{n!}\frac{d^n}{dx^n}(x^ne^{-x}), \qquad n=0,1\ldots
$$ 
 Using the  relation  ${T}{q}_n(2x)={q}_{n+1}(2x)$ 
 \cite[p. 363]{AG}, where $T$ is defined by \eqref{KKKK1}
 and taking \eqref{FFFF} into account  we arrive at the conclusion that ${\bf q}$ is orthogonal to the subspace 
 ${T}^{m+1}L_2({\mathbb R}_+)=\psi(\mathcal{B})L_2({\mathbb R}_+)$, where
 $\psi(\delta)=\left(\frac{\delta-i}{\delta+i}\right)^{m+1}$ belongs to ${H^\infty}(\mathbb{C}_+)$.
Hence, ${\bf q}$ is a non-cyclic function and for operators ${\bf H}_{a{\bf q}}$ there exist 
$S$-matrices defined by \eqref{ups4}. 

Let us calculate the $S$-matrix for the function
$q(x)=Me^{-|x|}.$ 
In this case, one can set $m=0$, $\psi(\delta)=\frac{\delta-i}{\delta+i}$, and $\Psi_1(z)=\Psi_2(z)=\left(\frac{z+i}{z-i}\right)^{2}$ 
as the holomorphic continuation of $\frac{\psi(-\delta)}{\psi(\delta)}=\left(\frac{\delta+i}{\delta-i}\right)^{2}$ into $\mathbb C_-$.
 Further,  
  $$
({\bf H}_\infty-z^2I)^{-1}{\bf e}^{-x}=\frac{{\bf e}^{-iz{x}}-{\bf e}^{-x}}{1+z^2}, \quad  W(z^2)=-2iz-\frac{4Re \ M}{1+iz}+\frac{|M|^2}{(1+iz)^2}.
$$
It follows from \eqref{AGH44b}  and the Poisson formula \cite[p.147]{Nik} that
$$
c(\mu, q_i)=\frac{\mu+i}{\mu-i}-\frac{M}{(\mu-i)^2}=\frac{\mu^2+1-M}{(\mu-i)^2}.
$$
After substitution of the expressions above into \eqref{ups4} and elementary transformations we find 
$$
S(z)=\left(\frac{z+i}{z-i}\right)^{2}\left(\sigma_0 -\frac{2iz(1-\frac{M}{z^2+1})(1-\frac{\overline{M}}{z^2+1})}{a-W(z^2)}
\left[ \begin{array}{cc}
1 &  1 \\
1 &  1
\end{array} \right]\right).
$$
Let us assume for the simplicity that $M\in{i}\mathbb{R}$. Then
\begin{equation}\label{hhh11}
S(z)=\left(\frac{z+i}{z-i}\right)^{2}\left(\sigma_0 -\frac{2iz(1+\frac{|M|^2}{(z^2+1)^2})}{a-W(z^2)}
\left[ \begin{array}{cc}
1 &  1 \\
1 &  1
\end{array} \right]\right)
\end{equation}
and $W(\lambda)=-2i\sqrt{\lambda}+\frac{|M|^2}{(1+i\sqrt{\lambda})^2}$, where $\lambda=z^2$ and $\sqrt{\lambda}=z$.

Since the first derivative of $W(\lambda)$ is
$$
W'(\lambda)=-\frac{i}{\sqrt{\lambda}}\left(1+\frac{|M|^2}{(1+i\sqrt{\lambda})^3}\right),
$$
 the equation $W'(\lambda)=0$ have the following roots $\lambda_j=z^2_j$, $j\in\{1, 2, 3\}$, where
$$
z_1=-\frac{\sqrt 3}{2}|M|^\frac{2}{3}+i(1-\frac{1}{2}|M|^\frac{2}{3}), \quad
z_2=-\overline{z_1}, \quad z_3=i(|M|^\frac{2}{3}+1).
$$

Assume that $|M|^2>8$. Then $z_1, \, z_2\in \mathbb C_- $.  Denote $a=W(z^2_1)$. Then the $S$-matrix \eqref{hhh11}
has a non-simple pole for $z=z_1$ and, by Lemma \ref{nnn3}, the operator ${\bf H}_{a{\bf q}}$ has an exceptional point $z_1^2$.
(The choice of $z_2=-\overline{z}_1$ instead of $z_1$ leads to the conclusion that the point $\overline{z}_1^2$ is exceptional for the adjoint operator
${\bf H}_{a{\bf q}}^*={\bf H}_{\overline{a}{\bf q}}$.)

The obtained result shows that the existence of exceptional points for some operators of the set 
$\{{\bf H}_{a{\bf q}}\}_{a\in\mathbb{C}}$, where ${\bf q}(x)=M{\bf e}^{-x}$, $M\in{i\mathbb{R}}$ depends on the absolute value of the imaginary $M$.
If $|M|^2>8$, then there exist two operators ${\bf H}_{a{\bf q}}$  and ${\bf H}_{\overline{a}{\bf q}}$ with the exceptional points
$z_1^2$ and $\overline{z}_1^2$ , respectively. On the other hand, if $|M|$ is sufficiently small  ($|M|^2\leq{8}$), then the collection
of operators $\{{\bf H}_{a{\bf q}}\}_{a\in\mathbb{C}}$ has no exceptional points.

\bigskip

\noindent\textbf{Acknowledgements}\\
\textit{This research  was partially supported by the Faculty of Applied Mathematics AGH UST statutory tasks within subsidy of 
Ministry of Science and Higher Education.}

\bigskip

\noindent Anna G\l\'{o}wczyk (corresponding author)\\  
glowczyk@agh.edu.pl\\

\noindent {\small
\noindent AGH University of Science and Technology\\
Faculty of Applied Mathematics AGH\\
Al. Mickiewcza 30, 30-059 Kraków
}\bigskip

\noindent Sergiusz Ku\.{z}el\\
kuzhel@agh.edu.pl\\

\noindent {\small
\noindent AGH University of Science and Technology\\
Faculty of Applied Mathematics AGH\\
Al. Mickiewicza 30, 30-059 Kraków 
}

\end{document}